\documentclass[11pt, aabib]{article}
\usepackage{natbib}

\usepackage{stix}

\usepackage{amsxtra, amsfonts, amssymb,latexsym}
\usepackage{amsthm,amsmath}
\usepackage{graphicx}
\usepackage[english]{babel}
\usepackage{accents}
\usepackage{hyperref}

\usepackage{pdfsync}

\theoremstyle{definition}

\newtheorem{theorem}{Theorem}[section]
\newtheorem{lemma}[theorem]{Lemma}

\newtheorem{proposition}[theorem]{Proposition}

\theoremstyle{definition}

\theoremstyle{remark}

\newcommand{\half}{\mbox{$\frac{1}{2}$}}
\newcommand{\beq}{\begin{equation}}
\newcommand{\eeq}{\end{equation}}

\newcommand{\as}{\mathrm{asym}}
\newcommand{\nLL}{n\mathrm{LL}}
\newcommand{\LLL}{\mathrm{LLL}}
\newcommand{\rv}{\mathbf{r}}

\newcommand{\qv}{\mathbf{q}}

\newcommand{\im}{\mathrm{i}}

\newcommand{\R}{\mathbb{R}}

\begin{document}
\centerline{\bf \Large\Large Quantum Hall states in higher Landau levels}
\bigskip\bigskip
{\large \bf Jakob Yngvason}

{\small Fakultät für Physik, Universität Wien, Boltzmanngasse 5, 1090 Vienna, Austria\\
Erwin Schrödinger Institute for Mathematics and Physics, Boltzmanngasse 9, 1090 Vienna, Austria.\\
E-mail: jakob.yngvason@univie.ac.at}

\bigskip

\bigskip\bigskip
\noindent{\bf \large Abstract} 

The unitary correspondence between Quantum Hall states in higher Landau levels and states in the lowest Landau level is discussed together with the resulting transformation formulas for particle densities and interaction potentials. This correspondence leads in particular to a representation of states in arbitrary Landau levels in terms of holomorphic functions. Some important special Quantum Hall states  in higher Landau levels are also briefly discussed.


\section{Key objectives}
\begin{itemize}
\item Describe the emergence of Landau levels from for the quantized  two-dimensional motion of charged particles in a perpendicular magnetic field.
\item Introduce non-commutative guiding center and cyclotron variables  and the associated two commuting harmonic oscillators.
\item Discuss the description of states in higher Landau levels in terms of holomorphic wave functions in the lowest Landau level.
\item Derive formulas that express particle densities and interaction potentials in an arbitrary Landau level in terms of corresponding quantities in the lowest Landau level.
\item Discuss briefly Laughlin states in higher Landau levels and recent investigations of states with filling factor 5/2.
\end{itemize}

\section{Notations and acronyms}

\noindent 
$\ell_B$: Magentic length\\
$\LLL$: Lowest Landau level (single particle) \\
$\LLL^N$: Lowest Landau level ($N$ particles) \\
$\nLL$: $n$-th Landau level (single particle) \\
$\nLL^N$: $n$-th Landau level ($N$ particles)\\
$\Pi_0$: Projector on LLL\\
$\Pi_n$: Prpjector on $\nLL$ 

\section{Introduction}

The quantum states of charged particles moving in a plane orthogonal to a homogeneous magnetic field are naturally grouped into Landau levels which correspond to quantization of the cyclotron motion of the particles  in the magnetic field. The interplay of this motion with external electric fields and  impurities as well as Coulomb interactions between the particles is the subject of Quantum Hall Physics which has developed into a major subfield of condensed matter physics since the discovery of the integer and fractional Quantum Hall Effects  in \cite{Kli} and \cite{StoTsuGos-82} respectively. These developments are extensively discussed in many monographs and review articles, see, e.g., \cite{ChaPie-88, ChaPie-95, Goerbig-09, GoeLed-06, Grei-2011, Jain-07, PraGir-1987, Tong-16}. The present chapter focuses on one particular theoretical aspect, namely the relations between some basic properties of Quantum Hall states in different Landau levels.

In strong magnetic fields the lowest Landau level, corresponding to the smallest of the discrete values for the cyclotron radius, plays a special role. The fractional quantum Hall effect was first discovered in strongly correlated states within this level and the  celebrated wave function introduced by \cite{Laughlin-83} for a quantum liquid at filling factor 1/3 (more generally $1/q$ with small odd values of $q$), turned out to be very successful in describing ground state properties as well as excitations.

An important mathematical feature of wave functions in the lowest Landau level is the fact that they form a Bargmann space of {\em holomorphic} functions when the position coordinates are expressed as complex numbers in the symmetric gauge. The holomorphy has been used for derivations of some basic properties of the states, for examples bounds on the local particle density for states related to the Laughlin wave function (\cite{LieRouYng-16, LieRouYng-17}). Wave functions in higher Landau levels, on the other hand,  involve also powers of the complex conjugate position variables in the standard representation and at first sight it might appear that the advantage of  holomorphic representatons gets lost.

It was, however, noted early \cite{MacDonald-84} that, due to a unitary correspondence between states in different Landau levels, wave functions in higher Landau levels can also be represented by holomorphic functions. In particular, there are Laughlin states in any Landau level. The physical picture behind this correspondence is 
based on the fact that the position variable has a natural decomposition into two sets of non-commutative variables, one associated with the cyclotron motion and the other with the  guiding centers around which the cyclotron motion takes place. Mathematically, these correspond to two sets of mutually commuting harmonic oscillators.  For a given Landau level, the cyclotron motion can be thought of as fixed, while the physics takes place in the guiding center degrees of freedom where holomorphic wave functions appear naturally. 

In this chapter the mathematical implementation of this idea is presented in two versions, one employing ladder operators and another one based on coherent states. Each brings different aspects of the unitary correspondence 
between states in higher Landau levels with those in the lowest Landau level into focus. Furthermore we derive formulas that relate particle densities and effective potentials  in higher Landau levels to their counterparts in the lowest Landau level.

In the last section we discuss some important special states in higher Landau levels and their relation to states in the lowest Landau level,

\section{The magnetic Hamiltonian and its Eigenfunctions}

\subsection{The Landau levels}

The magnetic Hamiltonian of a single spinless (or spin-polarized)  particle of charge $q$ and effective mass $m^*$, moving in a plane with position variables  ${\mathbf r}=(x,y)$, is
\beq H=\frac 1{2m^*}(\pi_x^2+\pi_y^2)\label{magnham}\eeq
where 
\beq \mbox{\boldmath$\pi$}=(\pi_x,\pi_y)=\mathbf p-q\mathbf A\eeq
is the gauge invariant kinetic momentum with 
\beq\mathbf p=-\im \hbar (\partial_x,\partial_y)\eeq
the canonical momentum and $\mathbf A$ the magnetic vector potential. We assume a homogeneous magnetic field of strength $B$ perpendicular to the plane and choose the
symmetric gauge  
\beq\mathbf A=\frac B2(-y,x).\eeq
Moreover, we choose units and signs such that $|q|=1$, $qB\equiv B>0$, $\hbar=1$ and $m^*=1$. Then
\beq \pi_x=-\im \partial_x+\half B y,\quad \pi_y=-\im \partial_y-\half B x\eeq
and the kinetic momentum components satisfy the canonical commutation relations (CCR)
\beq [\pi_x,\pi_y]=\im \ell_B^{-2}\label{CCRpi}\eeq
with 
\beq\ell_B=B^{-1/2}\eeq
the {\em magnetic length}.

In terms of the  creation and annihilation operators
\beq \quad a^\dagger=\frac {\ell_B}{\sqrt 2}(-\pi_y-\im \pi_x), \quad a=\frac{ \ell_B} {\sqrt 2 }(-\pi_y+\im \pi_x)\label{a}\eeq
with
 $[a,a^\dagger]=1$
one can write \eqref{magnham} as
\beq H=2B(a^\dagger a+\half) \eeq
which is the Hamiltonian of a harmonic oscillator with eigenvalues 
\beq \epsilon_n=(n+\half)2B,\,\, n=0,1,\dots.\label{Landauspec}\eeq
The Gaussian
\beq \varphi_0({\bf r})=\frac 1{\sqrt {\pi}}\,e^{-(x^2+y^2)/4\ell_B^2}\label{gaussian}\eeq
with $a \phi_0=0$ is a ground state for $H$ with energy $\epsilon_0=B$.

The discrete energy values \eqref{Landauspec} result from the  quantization 
of the classical cyclotron motion of a charged particle around {\em guiding centers} as discussed below. Every energy eigenvalue is infinitely degenerate due to the different possible positions of the guiding centers. The  degeneracy per unit area is $(2\pi \ell_B^2)^{-1}$.

The eigenspace in $L^2(\mathbb R^2)$ to the eigenvalue $\epsilon_n$ is called the {\em $n$-th Landau level }and denoted by $\nLL$. The {\em lowest Landau level},  corresponding to $n=0$ will be denoted by $\LLL$.
The corresponding fermionic spaces for $N$ electrons,   denoted by $\nLL^N$ and $\LLL^N$ respectively, are the antisymmetric tensor powers
\begin{equation}\label{eq:LLN}
\LLL^N =  \LLL^{\otimes_a N}, \quad \nLL^N =  \nLL^{\otimes_a N}.  
\end{equation}
For bosons the antisymmetric tensor product is replaced by the symmetric one; this is, e.g., relevant for cold atomic gases in rapid rotation, where the rotational velocity takes over the role of the magnetic vector potential \cite{Cooper}.

The {\em filling factor} $\nu$  in a given Landau level is by definition the ratio of the particle density in that level to the degeneracy $(2\pi \ell_B^2)^{-1}$ per unit area. In a sample with area $A$ a completely filled Landau level of fermions (filling factor $\nu=1$) thus contains $N= (2\pi \ell_B^2)^{-1} A $ particles.

{\bf Remark.} In the whole chapter, with the exception of the last Section 8.2,  we focus on the orbital motion of the charged particles, ignoring spin and the accompanying  Zeeman energy $\sum_i g^*\mu_B\mathbf B\cdot\mathbf S_i$ (with $\mathbf S_i$ the spin operator for particle $i$, $\mathbf B$ the magnetic field, $\mu_B$ the Bohr magneton and $g^*$ the effective Land\'{e} g-factor). This is justified if the particles are either spinless, or their spin is fully polarized due to a strong magnetic field. Moreover,  in semiconductors  the  Zeeman energy is typically much smaller than the gap between Landau levels because the effective g-factor is small. See, e.g.,  \cite{Tong-16}, p.~16.

\subsection{The two oscillators}\label{sec:guiding}
Quantum mechanically the dynamics of the guiding centers is described by another harmonic oscillator commuting with the cyclotron oscillator \eqref{magnham}. While the latter determines the Landau energy spectrum \eqref{Landauspec}, the quantization of the guiding center oscillator parametrizes the degeneracy creating a natural basis of states in each Landau level.

 One arrives at this  picture by splitting the (gauge invariant) position operator $\bf r$ into a guiding center part $\bf R$ and the cyclotron part 
\beq \widetilde{\mathbf R}=\ell_B^2\mathbf n\times \mbox{\boldmath$\pi$},\eeq 
with $\mathbf n$  the unit normal vector to the plane. Both  $\bf R$  and $\widetilde{\bf R}$ are gauge invariant and they commute with each other. On the other hand the two components of $(R_x, R_y)$ of $\bf R$  do not commute  and likewise for the components of $\widetilde{\bf R}$. More precisely, we have
\beq \mathbf r=\mathbf R+\widetilde{\mathbf R}\label{splitting}\eeq
with
\beq R_x=x+\ell_B^2 \pi_y=\half x-\im \ell_B^2\partial_y, \quad R_y=y-\ell_B^2\pi_x=\half y+\im \ell_B^2\partial_x,\eeq
\beq \widetilde{R}_x=-\ell_B^2 \pi_y=\half x+\im \ell_B^2\partial_y, \quad \widetilde{R}_y=\ell_B^2\pi_x=\half y-\im \ell_B^2\partial_x\eeq
and the commutation relations
\beq[\mathbf R,\widetilde {\mathbf R}]=\mathbf 0,\quad
[R_x,R_y]=-\im \ell_B^2,\quad [\widetilde{R}_x,\widetilde{R}_y]= \im \ell_B^2.\label{CCR}
\eeq
The creation and annihilation operators for $\widetilde{\mathbf R}$ are the same as \eqref{a}, namely
\beq a^\dagger=\frac1{\sqrt 2 \ell_B}(\widetilde{R}_x-\im \widetilde{R}_y),\quad  a=\frac1{\sqrt 2 \ell_B}(\widetilde{R}_x+ \im \widetilde{R}_y).\label{aa}\eeq
Those for the guiding centers, on the other hand, are
\beq b^\dagger=\frac 1{\sqrt 2 \ell_B}(R_x+\im R_y),\quad b=\frac 1{\sqrt 2 \ell_B}(R_x-\im R_y).\label{b}\eeq
with 
$[b,b^\dagger]=1$ and $[a^{\#}, b^{\#}]=0$.
Note the different signs compared to \eqref{aa} due to the different signs in \eqref{CCR}. 

\subsection{Complex notation}\label{sec:complex}

The two-dimensional configuration space $\mathbb R^2$ can be identified with the complex plane $\mathbb C$. Defining complex coordinates and derivatives by
\begin{equation}
z =x+\im y, \quad \bar z=x-\im y, \quad \partial_z=\half(\partial_x-\im \partial_y), \quad \partial_{\bar z}=\half(\partial_x+\im \partial_y) 
\end{equation}
we can write
\beq a^\dagger=\frac 1{\sqrt 2 \ell_B} (\half \bar z-2\ell_B^2 \partial_z),\quad   a=\frac 1{\sqrt 2 \ell_B} (\half z+2\ell_B^2 \partial_{\bar z}).\label{aellb}\eeq
Choosing units so that $B=2$, or equivalently, defining $z=\frac 1{\sqrt 2 {\ell_B}}(x+\im y)$, this simplifies to
\beq a^\dagger=\half \bar z- \partial_z, \quad  a= \half z+ \partial_{\bar z}.\label{a-operators}\eeq
Also, the gaussian factor
$e^{-(|x|^2+|y|^2)/4\ell_B^2}$ 
becomes $e^{-|z|^2/2}$.
We note that besides the standard definition $z=x+\im y$, other complexifications of $\mathbb R^2$ are  possible and can be useful, as stressed in~\cite{Haldane-18}.

For computations it is often convenient to use instead of \eqref{a-operators} the operators  $\hat a^\dagger$, $\hat a$ which act only on the pre-factors to the gaussian and are defined by
\beq  a^{\#} \left[f(z,\bar z)e^{-|z|^2/2}\right]=\left[\hat a^{\#} f (z,\bar z)\right]e^{-|z|^2/2}.\eeq
These are
\beq \hat a^\dagger=\bar z-\partial_z, \quad  \hat a=\partial_{\bar z}.\label{hata}\eeq
In the sequel we shall generally use the hat $\hat{\phantom a}$ on operators and functions to indicate that 
gaussian normalization factors are dropped from the notation.

The creation and annihilation operators for the guiding center oscillator are in complex notation

\beq   b^\dagger=\frac 1{\sqrt 2 \ell_B} (\half z-2\ell_B^2 \partial_{\bar z}),\quad b=\frac 1{\sqrt 2 \ell_B} (\half \bar z+2\ell_B^2 \partial_z).\label{bellb}\eeq
For $B=2$ 
\beq b^\dagger=\half z- \partial_{\bar z}, \quad  b= \half \bar z+ \partial_{z}\eeq
and 
\beq   \hat b^\dagger=z-\partial_{\bar z},\quad \hat b=\partial_{z}.\label{hatb}\eeq
The splitting \eqref{splitting} corresponds to
\beq (z,\bar z)=(b^\dagger, b)+(a,a^\dagger).\label{splitting2}\eeq
While the operators $a^\dagger, a$ increases or decrease the Landau level index, the operators $b^\dagger, b$ leave each Landau level invariant. Pictorially speaking we can say that operators $a^\#$ associated with the cyclotron oscillator are ladder operators moving states \lq\lq vertically\rq\rq\ on the energy scale,  while the ladder operators $b^\#$  associated with the guiding center oscillator move states \lq\lq horizontally\rq\rq, keeping the energy fixed. 

\subsection{Eigenfunctions}

With the gaussian  $\varphi_{0,0}:=\varphi_0$ \eqref{gaussian} as the common, normalized lowest energy state for the two oscillators the states 
\beq\varphi_{n,m}= \frac1{\sqrt{n!m!}}(a^\dagger)^n(b^\dagger)^m\varphi_{0,0}=\frac1{\sqrt{n!m!}}(b^\dagger)^m(a^\dagger)^n\varphi_{0,0},\quad n,m=0,1,\dots\label{eigenfunctions}\eeq
form a basis of  eigenstates of the Hamiltonian \eqref{magnham}. For fixed $n$ the states $\varphi_{n,m}$, $m=0,1,\dots$ generate the Hilbert space $\nLL $ of the $n$'th Landau level with $n=0$ corresponding to the lowest Landau level $\LLL$.

In complex coordinates the wave functions with $n=0$ respectively $m=0$, are
\beq \varphi_{0,m} (\rv)=\frac 1{\sqrt{\pi m!}}\, z^m e^{-|z|^2/2},\quad \varphi_{n,0}(\rv)=\frac 1{\sqrt {\pi n!}}\, \bar z^ne^{-|z|^2/2}.\eeq
More generally, one can write
\beq \varphi_{n,m}(\rv)=\frac 1{\sqrt {\pi n!\,m!}} [(z-\partial_{\bar z})^m \bar z^n ] e^{-|z|^2/2}=\frac 1{\sqrt {\pi n!\,m!}}[(\bar z-\partial_{z})^n z^m ] e^{-|z|^2/2}\label{nLLbasis}.\eeq
These functions  can be written in polar coordinates in terms of associated Laguerre polynomials \cite{LaLi,Gra-94}. They are simultanous eigenfunctions of the Hamiltonian \eqref{magnham} and the angular momentum operator $\mathbf L=\mathbf r \times \mathbf p$ in the symmetric gauge. Its action on the pre-factor of the gaussian is
\beq \hat L=z\partial_z-\bar z\partial_{\bar z}\eeq
with eigenvalues $M=m-n,\ m=0, 1, \dots $ in the $n$LL. The operators $b^\dagger$ and $b$ shift the angular momentum within each Landau level. 

\subsection{Factorization}

The mutual commutativity of the two oscillators has an important consequence:

\begin{lemma}[{\bf Factorization}]\label{lem:factorization}
{\em If $A$ is a function of $a^\dagger, a$ and $B$ of $b^\dagger, b$, then
\beq
\langle \varphi_{n',m'}|AB|\varphi_{n,m}\rangle=\langle \varphi_{n',0}|A|\varphi_{n,0}\rangle\, \langle \varphi_{0,m'}|B|\varphi_{0,m}\rangle.\label{factorization}
\eeq}
\end{lemma}

The proof follows from the fact that the two commuting harmonic oscillators \eqref{a} and \eqref{b} can be represented, in a unitarily equivalent way,  in the tensor product of two spaces with basis vectors $\varphi_{n,0}$ and $\varphi_{0,m}$ respectively. In this representation $\varphi_{n,m}\simeq \varphi_{n,0}\otimes \varphi_{0,m}$ and the operators $A$ and $B$  act independently on each of the tensor factors. (Alternatively,  one can choose $A,B$ to be polynomials in the creation and annihilation operators and use the CCR to prove \eqref{factorization} directly.)
Note, however, that in the representation \eqref{nLLbasis} the functions $\varphi_{n,m}(z,\bar z)$ are not simply products of the functions $\varphi_{n,0}(z)$ and $\varphi_{0,m}(\bar z)$. Indeed, the  variables $z$ and $\bar z$ do not act independently in the tensor factors because they involve sums of $a^\#$'s and $b^\#$'s , cf.~Eq. \eqref{splitting2}. The lemma, however, is in accord with  the intuitive picture, stressed in \cite{Haldane-13, Haldane-18}, that
in a given Landau level all the non-trivial structure is in the guiding center degrees of freedom while 
 the cyclotron motion just produces a decoration by a fixed factor in matrix elements.
 
 \subsection{The lowest Landau level as a Bargmann space of holomorphic functions}

Wave functions $\psi$ in the  lowest Landau level  $\LLL$ are characterized by the equation
\beq a\psi=0\eeq
which by \eqref{hata} means that 
\beq \psi(\rv)=f(z)e^{-|z|^2/2}\eeq
with $f$ satisfying the Cauch-Riemann equation $\partial_{\bar z}f=0$, i.e., $f$ is a {\em holomorhic} function of $z$. Holomorphic functions that are square integrable on $\mathbb R^2$ with respect to the Lebesgue measure with the gaussian weight factor $e^{-|z|^2}$ form a {\em Bargmann space} (\cite{Bargmann-61}).

In the same way,  $N$-particle fermionic wave functions 
in $\LLL^N$ have the form
\beq \Psi(\rv_1\dots, \rv_N)=F(z_1,\dots, z_N)e^{-(|z_1|^2+\cdots +|z_N|^2/2}\eeq

with a holomorphic, antisymmetric function $F$ of the $N$ variables $z_1,\dots, z_N$.

In the sequel we shall show that representations of states by holomorphic functions are not limited to the lowest Landau level. 

\section{Unitary maps between Landau levels}

The structure of the eigenfunctions  \eqref{eigenfunctions} with the mutually commuting sets of  creation and annihilation operators $a^\#$ and $b^\#$ for the two oscillators leads immediately to a unitary correspondence $n$LL $\leftrightarrow$ LLL.
In fact, the operator
\beq U_n= (n!)^{-1/2} a^n\quad\hbox{restricted to $n$LL}\label{Un}\eeq
maps $\varphi_{n,m}\mapsto \varphi_{0,m}$ because $a$ commutes with $b^\dagger$, and the inverse
\beq U_n^{-1}= (n!)^{-1/2} (a^\dagger)^n\quad\hbox{restricted to LLL}\label{Un-1}\eeq 
takes $\varphi_{0,m}\mapsto\varphi_{n,m}$. Hence $U_n$ is a unitary map $n$LL $\to$ LL.

The wave functions in $\nLL$ are by \eqref{nLLbasis} polynomials in $\bar z$ of order $n$ with coefficiants that are holomorphic functions of $z$. Using the representations \eqref{hata} for the creation and annihilation operators we conclude that the following holds:  

\begin{proposition}[\textbf{Unitary map between $\nLL$ and $\LLL$}]\label{lem:simple}\mbox{}\\
{\em Let $\psi_n\in\nLL$ have wave function
\beq\label{eq:nLLfunc}
\psi_n(\rv)=\sum_{k=0}^n \bar z^k f_k(z)e^{-|z|^2/2},\eeq 
with $f_k$ holomorphic for $k=0,\dots, n$. Then $\psi_0 = U_n \psi_n\in\LLL$ has the wave-function
\beq
\psi_0(\rv)= \sqrt{n!}f_n(z)e^{-|z|^2/2} \label{40}.
\eeq
Conversely, the wave function of $\psi_n = U_n^{-1}\psi_0\in\nLL$ is
\begin{align}\label{mainformula} 
\psi_n(\rv) &=  [(\bar z-\partial_z)^nf_n(z)]e^{-|z|^2/2}\nonumber\\
&= \left[\bar z^n f_n(z)+\sum_{k=1}^n(-1)^k {n\choose k} \,\bar z^{n-k}\partial^kf_n(z)\right]e^{-|z|^2/2}.
\end{align}}
\end{proposition}

Note that Equation~\eqref{mainformula} implies in particular that the holomorphic factor $f_n(z)$ to the highest power $n$ of $\bar z$ determines uniquely the factors to the lower powers $\bar z^k$:
\beq f_k(z)= (-1)^{n-k}\left({\begin{array}{c} n\\ k\\ \end{array}}\right) f_n^{(n-k)}(z).\label{134}\eeq
A wave function in $\nLL$  is thus {\em completely fixed by the holomorphic function $f_n$ and the Landau level index $n$}.
\bigskip

These considerations lead also to a formula for projecting functions into the lowest Landau level:

\begin{proposition}[\textbf{LLL projection}]\mbox{}\\
{\em Let $\phi$ be a state in $\bigoplus_{k=0}^n \: k\mathrm{LL}$ with wave function
\beq \phi(\rv)=\sum_{k=0}^n \bar z^k g_k(z)e^{-|z|^2/2}\label{135}\eeq
where the $g_k$ are holomorphic functions.  It's orthogonal projection into $\LLL$ is 
\beq\label{eq:LLL projection}
\Pi_0\phi(z)=\sum_{k=0}^n \partial^k g_k(z) e^{-|z|^2/2} .
\eeq}
\end{proposition}
This is known as the recipe ``move all $\bar{z}$ factors to the left and replace them by derivatives in $z$'', see e.g.~\cite{Jain-07}.
\begin{proof}
The previous considerations lead by induction to a splitting of  the state \eqref{135}
into its components in the different $k$LL, $k=0,\dots,n$: Start with a wave function as in~\eqref{135}. It's component $\psi_n$ in the $n$LL is  given by \eqref{mainformula} with $f_n:=g_n$ and $f_k$ for $0\leq k\leq n-1$ defined by~\eqref{134}. The difference
$\tilde \phi=\phi-\psi_n$ is now in $\bigoplus_{k=0}^{n-1} k\mathrm{LL}$ and we can repeat the procedure with $n$ replaced by $n-1$, $\phi$ by $\tilde\phi$  etc. until we obtain the splitting
$\phi=\sum_{k=0}^n \psi_{k}$
with $\psi_{k}\in k\mathrm{LL}$.
By induction over $n$, using that 
$
\sum_{k=0}^n {n\choose k} (-1)^k=(1-1)^n=0$,
this procedure implies~\eqref{eq:LLL projection}.
\end{proof}


\section{Many body states and $\ell$-particle densities}\label{densities}

The considerations of the previous sections carry in a straightforward manner over to many-body states in symmetric or anti-symmetric tensor powers $\nLL^N\equiv
 n$LL$^{\otimes_{{\rm s,a}}N}$ of single-particle states by applying the single-particle formulas to each tensor factor. 
 
Let $\Psi_n$ be a state in $\nLL^N$ with wave function 
\beq
\Psi_n(\rv_1, \dots, \rv_N)=\widehat{\Psi}_n(\rv_1, \dots, \rv_N)e^{-(|z_1|^2+\cdots +|z_N|^2)/2}.
\eeq
Expanding in powers of $\bar z_i$ we can write
\beq 
\widehat{\Psi}_n(\rv_1, \dots, \rv_N)=\prod_{i=1}^N \bar z_i^n f_n(z_1,\dots, z_N)+\sum\prod_{i=1}^N\bar z_i^{k_i} f_{k_1,\dots, k_N} (z_1,\dots, z_N).
\eeq
The sum is here over $N$-tuples $(k_1,\dots k_N)$ such that $k_i<n$ for at least one $i$. The functions $f_n$ and $f_{k_1,\dots, k_N}$ are holomorphic and the latter ones are, in fact, derivatives of $f_n$, cf. \eqref{134}.

Consider now the state $\Psi_0=\mathcal U_{N,n}\Psi_n$ in LLL$^N$ where 
\beq \mathcal U_{N,n}=U_n\otimes\cdots\otimes U_n\quad\hbox{with } U_n \hbox{ defined by \eqref{Un}}.\label{47}\eeq
Its wave function is
\beq \Psi_0(\rv_1, \dots, \rv_N)=\widehat \Psi_0(z_1,\dots, z_N)e^{-(|z_1|^2+\cdots +|z_N|^2)/2}\eeq
with the holomorphic function
\beq \widehat \Psi_0(z_1,\dots, z_N)=(n!)^{-N/2}\prod_{i=1}^N \partial_{\bar z_i}^n \, \widehat \Psi_n(z_1,\bar z_1; \dots ;z_N,\bar z_N)=
(n!)^{N/2}f_n(z_1,\dots, z_N).
\eeq 
Now $\Psi_n=\mathcal U_{N,n}^{-1}\Psi_0$,  so the  wave function $\widehat\Psi_n$ can by \eqref{Un-1} and \eqref{hata} be expressed in terms of the holomorphic function $f_n$:
\beq \widehat\Psi_n(z_1,\bar z_1; \dots ;z_N,\bar z_N)=(n!)^{-N/2}\prod_{i=1}^N (\bar z_i-\partial_{z_i})^n\widehat \Psi_0(z_1,\dots, z_N)=
\prod_{i=1}^N (\bar z_i-\partial_{z_i})^nf_n(z_1,\dots, z_N).\eeq

Next we consider {\em $\ell$-particle densities}, $\ell=1,2,\dots$,  which for a general
$N$-particle wave function $\Psi$ are defined by
\beq 
\rho^{(\ell)}_\Psi(\mathbf r_1,\dots, \mathbf r_\ell)= {N \choose \ell} \int_{\mathbb R^{2(N-\ell)}} |\Psi(\mathbf r_1\cdots \rv_\ell;\mathbf r'_{\ell+1}\cdots \mathbf r'_N)|^2\mathrm d\mathbf r'_{\ell+1}\cdots  \mathrm d\mathbf r'_N.
\eeq\smallskip

\begin{theorem}[\textbf{Particle densities in the $n$-th Landau level}]\label{thm:main2}\mbox{}\\
{\em The $\ell$-particle densities of $\Psi_n\in n{\rm LL}^N$ and  $\Psi_0 =\mathcal U_{N,n}\Psi_n\in {\rm LLL}^N$ are connected by 
\beq\label{eq:nLLdens}\rho^{(\ell)}_{\Psi_n}(\mathbf r_1,\dots, \mathbf r_\ell)
= \prod_{i=1}^\ell L_n \left(-\mbox{$\frac 14$}\Delta_{\mathbf r_i}\right)\rho^{(\ell)}_{\Psi_0}(\mathbf r_1,\dots, \mathbf r_\ell)
\eeq
where $L_n$ is the $n$-th Laguerre polynomial \beq L_n (t) = \sum_{l=0} ^n {n\choose l} \frac{(-t)^l}{l!}.\label{Laguerre}\eeq }
\end{theorem}
\smallskip

The proof is a simple consequence of the following Lemma:
\begin{lemma}[\textbf{Reshuffling differentiations}]\label{lem:proj dens}\mbox{}\\
{\em Let  $f$ be a holomorphic function. 
Then
\begin{equation}
(n!)^{-1}\overline {\left[(\bar z-\partial_{z})^n f(z)\right]}\left[(\bar z-\partial_{z})^n f(z)\right]e^{-z \bar z}= 
L_n\left(- \partial_{\bar z}\partial_z\right)\left[\overline{f(z)}f(z)e^{-z \bar z}\right].\label{nlift0} 
\end{equation}}
\end{lemma}
\begin{proof} This is a straightforward computation by induction over $n$, using the recursion relation for the Laguerre polynomials,
\beq (n+1)L_{n+1}(u)=(2n+1)L_n(u)-nL_{n-1}(u)-uL_n(u).\eeq
To compute \beq\partial_{\bar z}\partial_z\left[ \overline {\left[(\bar z-\partial_{z})^n f(z)\right]}\left[(\bar z-\partial_{z})^n f(z)\right]e^{-z \bar z}\right],\eeq
starting with $n=0$ and $L_0=1$, one uses the commutation relations
\beq
\partial_{\bar z} (\bar z-\partial_z)^n=(\bar z-\partial)^n\partial_{\bar z}+n(\bar z-\partial_z)^{n-1},\quad \partial_{z} (\bar z-\partial_z)^n=(\bar z-\partial)^n\partial_{z}, 
\eeq
and the fact that $\partial_{\bar z}f(z)=\partial_{z}\overline{f(z)}=0$ for holomorphic $f$. 
\end{proof}
 Theorem \ref{thm:main2} follows by applying the Lemma to the  first $\ell$ variables of $\Psi_0$,
 noting  that $\partial_{\bar z} \partial_z=\frac 14 \Delta$. 

\section{Mapping Hamiltonians to the LLL}

Consider an $N$-body Hamiltonian of the form 
\beq H_{V,w}=\sum_{i=1}^N \left[H^{(i)} +V(\rv_i)\right ]+\sum_{i<j} w(\rv_i-\rv_j)\label{eq:full hamil}\eeq
where $H^{(i)}$ is the 1-particle Landau Hamiltonian \eqref{magnham} for particle $i$, $V$ is an external potential, and $w$ a translationally invariant two-body interaction potential. 

A common situation in Quantum Hall physics is that the first
$n-1$ Landau levels are filled with electrons and can be regarded as \lq\lq inert\rq\rq\ while the interesting phenomena take place in the tensor factor of the state in the $n$-th Landau level.
The formulas of the last section, which establish a unitary correspondence between states and densities in different Landau levels, make it possible to reduce the study of the Hamiltonian in the $n$-the level to a unitarily equivalent Hamiltonian in the lowest Landau level where states are described by holomorphic wave functions.  This reduction is accompanied by a transformation of the potentials $V$ and $w$. The precise statement is as follows:
\begin{theorem}[\textbf{Effective Hamiltonians in different LL's}] \label{thm:main}\mbox{}\\
{\em Let $H_{V,w}$ be given by~\eqref{eq:full hamil} and define
\beq H^{\nLL}_{V,w} = \Pi_n^N H_{V,w} \Pi_n^N \eeq
where $\Pi_n^N$ is the orthogonal projector from $L^2_{\as}(\R^{2N})$ to $\nLL^N$.
Then, for any $n$, 
\beq H^{\nLL}_{V,w}- n\cdot 2BN \eeq
is unitarily equivalent to 
\beq H^{\LLL}_{V^{(n)},w^{(n)}} = \Pi_0^N H_{V^{(n)},w^{(n)}}\Pi_0^N\label{61}\eeq
with $\Pi_0^N$  the projector onto the lowest Landau level $\LLL^N$  and the  effective potentials 
\begin{align} V^{(n)} (\rv) &= L_n \left(-\mbox{$\frac 14$}{\Delta} \right) V (\rv) \label{eq:eff pot}
\\ 
 w^{(n)} (\rv) &= L_n \left(-\mbox{$\frac 14$} {\Delta}\right) ^2 w (\rv)\label{eq:eff int},
\end{align}
where $\Delta$ is the Laplacian and $L_n$ the Laguerre polynomial \eqref{Laguerre}.
More generally, one can consider interactions beyond the two-body case: For an $\ell$ particle interaction the $L_n \left(-\mbox{$\frac 14$} {\Delta}\right) ^2$ in \eqref{eq:eff int} is replaced by a product of $L_n \left(-\mbox{$\frac 14$} {\Delta_i}\right)$, $i=1,\dots,\ell$.}
\end{theorem}

\begin{proof}
The proof follows directly from a comparison of expectation values of $H^{\nLL}_{V,w}$ in a state $\Psi_n\in n{\rm LL}^N$ with the corresponding state $\Psi_0=\mathcal U_{N,n}^{-1}\Psi_n \in {\rm LLL}^N$, using the transformation formula \eqref{eq:nLLdens} for one- and two-particle densities (more generally, $\ell$-particle densities). 
\end{proof}

{\bf Remark}. This is the basic result that allows to study interaction in higher LL's in terms of effective interactions in the LLL. It can be found in different guises in a number of sources, in particular~\cite{Haldane-18,CheBis-18,Tong-16,Jain-07, MacDonaldGirvin-86, CifQui-10}.  

\subsection{Another proof of Theorem 5.1.}

An alternative approach to~\eqref{eq:eff pot} and~\eqref{eq:eff int} is based on  the splitting \eqref{splitting} of the position variables in guiding centers and cyclotron motion and the factorization $\exp({\mathrm i}\mathbf q\cdot \mathbf r)= \exp({\mathrm i}\mathbf q\cdot \mathbf R)\exp({\mathrm i}\mathbf q\cdot  \tilde {\mathbf R})$. This factorization carries over to matrix elements by Lemma \ref{lem:factorization}. 

\begin{lemma}[\textbf{Plane waves projected in different Landau levels}]\label{lem:plane}\mbox{}\\
{\em For any $\mathbf{q}\in \R ^2$, identify $e^{\im \mathbf{q}\cdot \rv}$ with the corresponding multiplication operator on $L^2 (\R^2)$, where $\rv$ is the spatial variable. Let $\mathbf{R}$ be the guiding center operator defined in Sec.~\ref{sec:guiding}, $\Pi_n$ the orthogonal projector on $\nLL$ and $U_n:\nLL \rightarrow \LLL$ the unitary map \eqref{Un}.
Then, in the LLL,
\begin{equation}\label{eq:proj wave}
U_n \Pi_n e^{\im \mathbf{q}\cdot \rv} \Pi_n U_n ^{-1} = L_n \left(\frac{|\mathbf{q}|^2}{4}\right) e^{-\frac{|\mathbf{q}|^2}{8}} \Pi_0  e^{\im \mathbf{q}\cdot \mathbf{R}} \Pi_0 = L_n \left(\frac{|\mathbf{q}|^2}{4}\right)  \Pi_0  e^{\im \mathbf{q}\cdot \mathbf{r}} \Pi_0 
\end{equation}
with $L_n$  the Laguerre polynomial \eqref{Laguerre}.}
\end{lemma}
The equality of the left-hand and the right-hand sides of~\eqref{eq:proj wave} can be seen as a Fourier transformed version of~\eqref{eq:nLLdens} (with $\ell=1$). 

The proof of the Lemma is a simple computation using the commutation relations \eqref{CCR} and the Baker-Campbell-Hausdorff formula
\beq e^{X+Y}=e^X\,e^Y\,e^{-\half[X,Y]}\eeq
for operators such that $[X,[X,Y]]=[Y,[X,Y]]=0$. The computation can be found in \cite{RouYng-20}, p. 12 and also in
\cite{Jain-07} (proof of Theorem~3.2) and in~\cite{GoeLed-06}.
Equations~\eqref{eq:eff int}--\eqref{eq:eff pot} 
follow from \eqref{eq:proj wave}  by the Fourier decompositions 
\beq V (\rv) = \int_{\R^2} \widehat{V} (\qv) e^{\im \qv \cdot \rv} d\qv\eeq
and 
\beq w(\rv_1 - \rv_2) = \int_{\R^2} \widehat{w} (\qv) e^{\im \qv \cdot \rv_1} e^{-\im \qv \cdot \rv_2} d\qv \eeq
using the the Fourier representation of the Laplacian, $-\Delta = |\qv|^2.$


\section{Coherent state representations}

The considerations so far, in particular the holomorphic representation of states in any Landau level 
via the formulas \eqref{40} and \eqref{mainformula}, are based on the use of the ladder operators $a^\#$ which connect higher Landau levels with the LLL.
We now discuss an alternative approach which offers additional insights.  It relies on a representation of states in $n$LL in terms of coherent states for the guiding center oscillator.
\subsection{Coherent states}
A coherent state in $n$LL with parameter $Z\in\mathbb C$ associated with the guiding center oscillator is defined in a standard way (\cite{KlaSka-85, ComRob-21}) as
\beq |Z,n\rangle=e^{(Zb^\dagger -\bar Zb)}\varphi_{n,0}=\sum_{m=0}^\infty \frac {Z^m}{\sqrt {m!}}\, \varphi_{n,m} e^{-|Z|^2/2}\label{cohstate}.\eeq
The overlap of two coherent states is
\beq
\langle Z,n|Z',n'\rangle=\delta_{n,n'}  e^{(2\bar Z Z'-|Z|^2-|Z'|^2)/2}=\delta_{n,n'}  e^{-|Z-Z'|^2/2}\,e^{{\mathrm i }\,\text{Im}\,(\bar Z Z')}\eeq
and
\beq
\int |Z,n\rangle\langle Z,n|\,\frac{{\mathrm d}^2Z}\pi =\sum_{m=0}^\infty |\varphi_{n,m}\rangle\langle \varphi_{n,m}|=\Pi_n\label{nproj}
\eeq
is the projector on $n$LL, where ${\mathrm d}^2Z:=\hbox{$\frac{ \mathrm i }{2}$} {\mathrm d}Z\wedge {\mathrm d}\bar Z$ is the Lebesgue measure on the plane.


The coherent states lead directly to an interpretation of $n$LL as a Bargmann space of holomorphic functions of the coherent state variable $Z$: If $\psi\in$\,$n$LL then
\beq \widehat{\Psi}(Z):= \langle \bar Z,n|\psi\rangle e^{|Z|^2/2}=\sum_{m=0}^\infty\langle \varphi_{n,m}|\psi\rangle\frac {Z^m}{\sqrt {m!}}\, \label{analyt}\eeq
is analytic in $Z$ and 
\beq \Psi(Z,\bar Z)=\widehat{\Psi}(Z)e^{-|Z|^2/2}\label{analyt2}\eeq
has the same $L^2$ norm as $\psi$ because of \eqref{nproj}. Thus the map 
\beq U_n:\psi\mapsto \Psi\eeq
is isometric from the $n$LL to the  LLL. From the definition it is clear that 
\beq U_n\varphi_{n,m}=\varphi_{0,m}\quad \hbox{and}\quad
U_n|Z,n\rangle=|Z,0\rangle\label{74}\eeq
with the inverse
\beq U_n^{-1}\varphi_{0,m}=\varphi_{n,m}\quad \hbox{and}\quad
U_n^{-1}|Z,0\rangle=|Z,n\rangle.\label{79}\eeq
Thus $U_n$ is identical to the unitary defined by ladder operators in \eqref{Un} and \eqref{Un-1}.
In particular the 
holomorphic function \eqref{analyt} is mathematically identical to the function $\hat\psi_0(z)=\sqrt{n!} f_n(z)$ in \eqref{40} if $\psi=\psi_n$ and $Z$ is identified with the complex position variable $z=x+\im y$.
Note, however, 
that $Z$ is associated with the (non-commutative) components of the guiding center operator $\mathbf R$ rather than the (commutative)  position operator $\mathbf r$. 

{\bf Remark.} By the definition \eqref{analyt} $\Psi$ depends linearly on $\psi$; the alternative definition $\Psi(Z)=\langle\psi|Z,n\rangle$, that is used in some references, leads to an anti-unitary correspondence.

\subsection{Integral kernels}\label{sec:kernel}

The unitary map $\psi\mapsto \Psi$ defined by \eqref{analyt} can also be achieved using the integral kernel
\beq G_n(Z,\bar Z;z,\bar z)= \frac 1{ \sqrt {\pi n!}}\,(z-Z)^n e^{-(|Z|^2+|z|^2-2Z\bar z)/2}= \frac 1{ \sqrt {\pi n!}}\,(z-Z)^n e^{-|z-Z|^2/2}
e^{-{\mathrm i }\,\text{Im}\, (\bar z Z)}.\label{G}\eeq
Namely, we have 
 \beq \Psi(Z,\bar Z)=\int G_n(Z,\bar Z;z,\bar z)\psi(z,\bar z)\,{\mathrm d}^2z \label{psiPsi}.\eeq
 The inverse map is given by 
\beq \psi(z,\bar z)=\int \bar G_n(z,\bar z; Z,\bar Z)\Psi(Z,\bar Z)\,{\mathrm d}^2Z\label{Psipsi}
\eeq
with
\beq
\bar G_n(z,\bar z; Z,\bar Z)=\frac 1{ \sqrt {\pi n!}}\,(\bar z-\bar Z)^n e^{-(|Z|^2+|z|^2-2\bar Z z)/2}= \frac 1{ \sqrt {\pi n!}}\,(\bar z-\bar Z)^n e^{-|z-Z|^2/2}
e^{-{\mathrm i }\,\text{Im}\, (z\bar Z)}.\label{Gbar}\eeq
The (short) proof can be found in \cite{RouYng-20}, p.~8 and also (with a slightly different notation) in \cite{ChaFlo-07}, Eq.~(34).

Note that $G_n$ can be written as
\beq g_n(z-Z)\,e^{-2{\mathrm i }\,\text{Im}\, (\bar z Z)}\eeq
where $g_n$ is essentially concentrated  in a disc of radius $\sim \sqrt{ n+1}$ and the factor is a phase factor. Recall also that the length unit is $\sqrt 2\ell_B\sim B^{-1/2}$.

A further remark is that $G_0$ is the reproducing kernel for the Bargmann space, confirming again that in the LLL $\Psi$ and $\psi$ are the same function on $\mathbb C$ just with different names for the variables. The phase factor in $G_n$ is essential for this to hold.

\section{Recap of the different expressions for the unitary maps}

In the previous sections have displayed three equivalent ways to represent a state $\psi\in n$LL by holomorphic functions in Bargmann space:
\begin{itemize}
\item[{\bf 1.}]  Apply the differential operator $\partial_{\bar z}$ $n$-times to the pre-factor of the Gaussian, cf. Eqs. \eqref{Un} and  \eqref{Un-1}. Equivalently: Expand the pre-factor in powers of $\bar z$ and keep only the highest power. The inverse mapping, LLL $\to$ $n$LL, is achieved by applying the differential operator $(\bar z-\partial_z)^n$ to the analytic function representing the state in the LLL.
\item[{\bf 2.}] Take the scalar product 
 $\langle \bar Z,n|\psi\rangle $
with a coherent state, cf. \eqref{analyt}.
\item[{\bf 3.}] Use Equation \eqref{psiPsi} with the integral kernel \eqref{G}.
\end{itemize}
The three methods generalize in all cases to many body states in $\nLL^N\equiv
 n$LL$^{\otimes_{{\rm s,a}}N}$ by applying the prescriptions to each tensor factor.
 
Computationally  the first method is usually the simplest, but the two others are sometimes more illuminating from the physics point of view since they bring the pivotal role of the guiding center variables into focus.

\section{Some special states}

\subsection{Laughlin states}

Consider an $N$-body Hamiltonian \eqref{eq:full hamil} with a rotationally symmetric interaction potential $w$ but without an external potential. Its projection onto the LLL can be written as\footnote{Note that fermionic wave-functions do not see the even $m$ terms of~\eqref{eq: LLL hamil}.}
\beq
H_{w} ^{\LLL} =  \sum_{i<j} \sum_{m\geq 0} \left\langle \varphi_{0,m} | w | \varphi_{0,m} \right\rangle (|\varphi_{0,m} \rangle \langle \varphi_{0,m} |)_{ij}. \label{eq: LLL hamil}
\end{equation}
The matrix elements 
\beq w_m:=\left\langle \varphi_{0,m} | w | \varphi_{0,m} \right\rangle=\frac{1}{\pi m!} \int_{\mathbb R^2} |w(\mathbf r)| r^{2m}e^{-r^2}d^2\mathbf r\label{pseudopot}\eeq
are usually called ``Haldane pseudo-potentials''\footnote{ For very strong interaction potentials of range much smaller than the magnetic length $\sim B^{-1/2}$, in particular if there is a hard core, an expansion of the interaction in terms of moments as in \eqref{eq: LLL hamil} is not adequate. This situation is analysed in  \cite{SY-20} which generalizes the paper \cite{LS-09}. It is shown that in an appropriate scaling limit the pseudo-potential operators $|\varphi_m\rangle\langle \varphi_m|$ also emerge, but with renormalized pre-factors involving the scattering lengths of the interaction potentials in the different angular momentum channels, rather than expectation values as in  \eqref{eq: LLL hamil}.}, cf. \cite{Haldane-1983}. If the sum is truncated at $m=q-1$ for some natural number $q$ then the {\em Laughlin wave function}
\begin{equation}\label{eq:Laughlin}
\Psi^{{\rm Lau}}_N (z_1,\ldots,z_N) = c_N \prod_{i<j} (z_i-z_j) ^{q} e^{- \sum_{j=1} ^N |z_j| ^2/2} 
\end{equation}
is an exact ground state of \eqref{eq: LLL hamil}
($L^2$-normalized by the constant in front).

The Laughlin state in LLL$^N$ can also be written as
\begin{equation}
\Psi^{{\rm Lau}}_N = c_N \prod_{i<j} (b_i-b_j)^{q}\varphi_{0,0}^{\otimes N} 
\end{equation}
and its counterpart in $n$LL$^N$, obtained by applying the unitary transformation $\mathcal U_{N,n}^{-1}$ defined by \eqref{47}
to $\Psi^{{\rm Lau}}_N$, is
\begin{equation}
\Psi_{n,N}^{\rm Lau}=c_N\prod_{i<j}(b^\dagger_i-b^\dagger_j)^q \varphi_{n,0}^{\otimes N}=
c_N\prod_{i<j}(b^\dagger_i-b^\dagger_j)^q \prod_{i=1}^N\left[ (a^\dagger_i)^n\,\varphi_{0,0}\right]
\label{Laughn1}.
\end{equation}
Its wave function is (cf. Lemma~\ref{lem:simple})
\begin{equation}
\Psi_{n,N}^{\rm Lau}(\rv_1,\dots,\rv_N) = c_N\left[ \frac 1{(n!)^{N/2}}\prod_{i=1}^N\left(\bar z_i-\partial_{z_i}\right)^n\prod_{i<j}(z_i-z_j)^q\right] e^{-(|z_1|^2+\cdots+|z_N|^2)/2}.\label{Laughn2}
\end{equation}
This is, in \emph{electronic position variables}, the exact ground state of a Hamiltonian obtained by
\begin{itemize}
 \item Projecting the Hamiltonian~\eqref{eq:full hamil} onto $\nLL^N$.
 \item Unitarily mapping the result down to an effective Hamiltonian on $\LLL^N$ using Theorem~\ref{thm:main}.
 \item Neglecting the one-body potential $V^{(n)}$ and truncating the Haldane pseudo-potential series \eqref{pseudopot}  for the interaction potential $w^{(n)}$ at $m=q-1$. \end{itemize}
 Note that the pseudopotentials $w_m^{(n)}$ now depend on $n$ because $w$ is replaced by the effective interaction \eqref{eq:eff int}. 
 
 If the wave function  \eqref{Laughn2} is transformed back to LLL$^N$ using the representation of $\mathcal U_{N,n}^{-1}$ in terms of the integral kernel \eqref{G}, we obtain again the function \eqref{eq:Laughlin} but  with the $z_i$ replaced by the coherent state {\em guiding center
variables} $Z_i$. Since the integral kernel is essentially concentrated  in a disc of radius $\sim \sqrt{ n+1}\,\ell_B$, one can expect that for large $B$ and not too large $n$ the particle densities of \eqref{eq:Laughlin}
 and \eqref{Laughn2} are essentially the same in extended systems. This is indeed vindicated by the more detailed analysis mentioned below.

 According to Laughlin's plasma analogy~\cite{Laughlin-83} the density profile of \eqref{eq:Laughlin} is for large $N$ well approximated by a droplet of radius $(\ell N)^{1/2}$ and fixed density $(\pi q)^{-1}$,
\beq
\varrho^{\rm flat}_N(\rv) := \begin{cases}\frac{1}{\pi q} \mbox{ if } |z|\leq \sqrt {q N}\\
                               0 \mbox{ otherwise}.
                             \end{cases}
\eeq
By a rigorous mean-field analysis it can be proved \cite{RouSerYng-13b} that this approximation holds in the sense of averages over discs of radius $N^\alpha$ with $1/2>\alpha>1/4$. More generally, the $k$-particle densities are well approximated in this sense by the $k$-fold tensor power of the flat density if $N\to\infty$.
The more refined analysis of classical Coulomb systems in \cite{LebSer-16,BauBouNikYau-15, Ser-20} leads to an extension of this result down to mesoscopic scales $N^\alpha$ for all $\alpha>0$. Using the transformation formula \eqref{eq:nLLdens} and the fact that the Laguerre polynomal $L_n(t)$ is close to unity for small values of the argument then allows to conclude the same result for fixed $n$ and large $N$. 
The precise statement can be found as Theorem 8.1 in \cite{RouYng-20}.

In order to adapt to an external potential $V$ a natural modification of the Laughlin state is to insert  {\em quasi-holes}, i.e., to consider states of the form
\beq \Psi(z_1,\dots, z_N)= C_N' \Psi^{{\rm Lau}}_N (z_1,\ldots,z_N)\prod_{i=1}^N f(z_i)\eeq
where 
$f(z)=c \prod_{j=1}^M(z-a_j)$ with $a_j\in \mathbb C$.  In \cite{RouYng-17} it is discussed how, for suitable choice of the quasi-hole locations $a_j$,  such states minimize the potential energy in an external potential among all states in the LLL which vanish like $(z_i-z_j)^q$ when two points come close.
This result generalizes to higher Landau levels for large $N$ by an analogous arguments as for pure Laughlin states, the intuition being again that in higher Landau levels  position coordinates should be replaced by guiding center coordinates and the potential $V$  by the effective potential 
$V^{(n)}=L_n \left(-\mbox{$\frac 14$} {\Delta}\right) V$. 

\subsection{Filling factor $\nu=5/2$}

A   Quantum Hall state with a  clear signature for filling factor $\nu=5/2$ was first observed by \cite{Willettetal-87}. It is thought to consist of  {\em two} completely filled lowest Landau levels  (i.e., with $n=0$)  but with opposite spin polarizations,  combined with a spin polarized state with $n=1$ and half filling.  This state is unusual since the great majority of observed plateaus in the Quantum Hall conductance correspond to filling factors with an odd denominator as seemingly natural for fermions. Moreover, there is no sign of a  plateau in the Hall resistance at $\nu=1/2$ for $n=0$ alone.

Since its discovery the $\nu=5/2$ state has been extensively studied experimentally, see \cite{Willett-13} for a comprehensive review. It is also very interesting theoretically because quasi-hole excitations of the principal candidates for describing this state may show {\it  non-abelian} braid statistics.  In brief, the reason is that 
the manifold of states with $2k$ quasiholes is $2^{k-1}$ degenerate and the braiding of two quasiholes defines a unitary map on the space spanned by this manifold. The unitaries corresponding to different braidings do not commute in general.
An adequate discussion of these aspects would go well beyond the scope of the present chapter but a concise account with further references can be found in \cite{Tong-16} as well as in \cite{Grei-2011}. A detailed discussion of non-abelian anyons and potential applications for quantum computation is given  in \cite{Nayak-08}. The recent review of \cite{Ma-22} discusses both theoretical and experimental aspects at filling factor $\nu=5/2$.


The main  candidates for the quantum Hall state at  $\nu=5/2$ are a {\em Moore-Read state} (\cite{Moore-Read-91}) as well as its variant obtained by particle hole conjugation. A Moore-Read state in the LLL for  an even particle number $N$ and has the form
\beq\Psi_{\rm MR}(z_1, \dots,z_N)={\rm Pf}\left(\frac 1 {z_i-z_j}\right )\prod_{i<j}^N(z_i-z_j)^q \prod_{i=1}^N e^{-|z_i|^2/2}\label{Pfaffianstate}\eeq
with the Pfaffian
\beq {\rm Pf}\left(\frac 1 {z_i-z_j}\right)=\mathcal A\left\{ \frac 1 {z_1-z_2} \cdots  \frac 1 {z_{N-1}-z_N}  \right\}\eeq
where $\mathcal A$ stands for  antisymmetrization over all possible pairings of the coordinates.  It has filling fraction $1/q$ and is antisymmetric for 
{\em even} values of $q$. It is also holomorphic because  singular factors in the Pfaffian are compensated by  the factors $(z_i-z_j)^q$. In contrast to the Laughlin wave function and its quasi-hole descendants the Moore-Read wave functions do not vanish when two particles come together, but may vanish when more particles become coincident. In particular for $q=1$ it is a ground state of the formal 3-body Hamiltonian
\beq H=\sum_{i<j<k}\delta(z_i-z_j)\delta(z_i-z_k).\eeq
The anti-Pfaffian state (\cite{Lee-07}) is the particle-hole conjugate version of \eqref{Pfaffianstate}.
The concept of particle-hole conjugation is thoroughly discussed in \cite{Zirn-21}, see also \cite{Gir-84}. In the present case it means replacing the creation operators $b_i^*$ by the annihilation operators $b_i$ and applying the result to a fully occupied LLL state (the Fermi sea) rather than the gaussian ground state $\varphi_{0,0}^{\otimes N}$, cf.  Eq. (3.1) in \cite{Zirn-21}. The Fermi sea in the LLL is a Slater determinant and has a wave function as in \eqref{eq:Laughlin} with $q=1$. The explicit formula for the anti-Pfaffian state,  given by \cite{Lee-07},  is 
\begin{multline}\label{antiPfaff}\Psi_{\rm aPf}(z_1, \dots,z_N)=\prod_{i<j}^N(z_i-z_j) \prod_{i=1}^N e^{-|z_i|^2/2}
\times \\ \int \prod_{\alpha=1}^N\mathrm d^2\eta_\alpha\prod_{i,\alpha}(z_i-\eta_\alpha)e^{-|\eta_\alpha|^2/2}\prod_{\beta<\gamma}(\eta_\beta-\eta_\gamma)\Psi_{\rm MR}(\bar \eta_1,\dots\bar\eta_N).\end{multline}
Yet another candidate for the quantum Hall state at $\nu=5/2$ is a particle-hole symmetric state suggested in \cite{Son-15}.  

All the states mentioned are specified in the lowest Landau level, where they are given by holomorphic wave functions. The general procedure of Sections 3-6 then takes care of lifting them up to any Landau level  $\nLL^N$. In the $\nu=5/2$ case one takes $n=1$ and assumes a fully spin polarized state. An important point is now  that the effective interaction  depends on the Landau level as described in 
Eq.\ \eqref{eq:eff int}. For instance, a  Coulomb potential is more repulsive at short distances for $n=1$ than for $n=0$, cf.\ Fig.\ 1 in \cite{CifQui-10}. This may explain why, e.g. a Moore-Read state with $q=2$ can have an energy gap in the $n=1$ level and exhibit a FQHE plateau, while being a gapless Fermi liquid in the LLL.  

The nature of the true ground state at $\nu=5/2$ is, however,  still a matter of considerable debate. The issues of spin polarization and the possibility of {\em Landau level mixing}  play an important role here. In this chapter so far it has  always been assumed that a single spin polarized Landau level, namely the highest partially filled one, is the arena for the phenomena under consideration while all lower levels are completely filled and can be regarded as inert.  Landau level mixing means that the projection of the interaction Hamiltonian into a single Landau level is not adequate and processes where interactions induce transitions between different levels become relevant. Mathematically this means that it is not sufficient to consider states that are simply a tensor product of a state in a partially filled level with filled lower lower levels but one must consider states which are linear combinations of states with partial filling in more than one level.  

Effective Hamiltonians to account for Landau level mixing at $\nu=5/2$  were proposed in \cite{Pet-13}, \cite{Sim-13} and \cite{Sod-13},   and in \cite{Rez-17} it was concluded that when mixing is allowed, an anti-Pfaffian state is favoured. The effects of mixing may depend, however, on the size and shape of samples, and clear-cut answers for the thermodynamic limit are hard to come by. Moreover, the
 strength of the mixing depends on the material with, e.g., ZnO showing stronger effects that  GaAs. The paper by \cite{Luo-21} deals in particular with systems with strong mixing and the question of stability for incompressible ground states.

\section{Conclusions} In this chapter it has been explained how quantum Hall states for many-body systems  in arbitrary Landau levels can be mapped unitarily onto states in the lowest Landau level. The basis for this mapping is a splitting of the two-dimensional  position variables into two sets of non-commutative variables associated respectively with the  guiding centers and the cyclotron motion of the particles in the magnetic field. In a fixed Landau level the cyclotron variables play the role of  bystanders while the state for the guiding center variables can be described by holomorphic wave functions independently of the Landau level. As a consequence, it is possible to treat particle densities and interactions in higher Landau levels in terms of suitably transformed densities and interactions in the lowest Landau level. 

These insights are applied to Laughlin states in arbitrary Landau levels as well as to Pfaffian and anti-Pfaffian states that have been proposed for explaining pronounced quantum Hall plateaus at filling fraction 5/2 and may exhibit non-abelian braid statistics of excitations. Their discussion in this chapter is quite brief, however, but detailed treatments of these matters are presented in other chapters of the Encyclopedia.

\noindent\textbf{Acknowledgements.} Thanks are due to Nicolas Rougerie for collaboration on the review  \cite{RouYng-20} which is the basis for the present chapter.

\bibliographystyle{haabib}

\end{document}